\newif\iflong
\longtrue

\documentclass[runningheads]{llncs}
\usepackage[T1]{fontenc}
\usepackage{graphicx}

\usepackage[linesnumbered,ruled,vlined]{algorithm2e}
\usepackage{hyperref}
\usepackage{amsmath}
\usepackage{amssymb}
\usepackage{complexity}
\usepackage{lmodern}
\usepackage{microtype}
\usepackage[all, defaultlines=3]{nowidow}
\usepackage{tcolorbox}
\usepackage{todonotes}
\usepackage{xspace}
\usepackage{cleveref}
\usepackage[]{xcolor}

\spnewtheorem{observation}{Observation}{\bfseries}{\itshape}

\Crefname{observation}{Observation}{Observations}

\newcommand{\lncsqed}{\hfill$\Box$}

\newcommand{\Dmc}{\ensuremath{\mathcal{D}}\xspace}
\newcommand{\Nbb}{\ensuremath{\mathbb{N}}\xspace}

\newcommand{\DE}{\textsc{Decoder Extension}\xspace}

\newcommand{\mpp}{\mathcal{P}}

\begin{document}
\title{Towards Settling the Complexity of the Lettericity Problem}
\author{Mario Grobler\inst{1}\orcidID{0000-0001-8103-6440}\and Nils Morawietz\inst{2,3}\orcidID{0000-0002-7283-4982} \and
Silas Cato Sacher\inst{4}\orcidID{0009-0004-6850-1298}}
\authorrunning{Grobler, Morawietz, and Sacher}
\institute{University of Bremen
\email{grobler@uni-bremen.de}\\
\and LaBRI, CNRS, University of Bordeaux, Bordeaux INP, France. 
\and Institute of Computer Science, Friedrich Schiller University Jena,  Germany
\email{nils.morawietz@uni-jena.de}\\
\and Universit\"at Trier, Fachbereich IV, Informatikwissenschaften, Germany
\email{sacher@uni-trier.de}}
\maketitle              
\begin{abstract}
The \emph{lettericity} of a graph $G=(V,E)$ is defined as the smallest size of an alphabet $\Sigma$ such that there is a word $w_1 \dots w_{|V|} \in \Sigma^*$ and a decoder $\mathcal{D} \subseteq \Sigma^2$ with the property that $G$ is isomorphic to the \emph{letter graph} $G(\mathcal{D}, w)$, that is, the graph with vertex set $\{1, \dots, n\}$ and edge set~$\{ij \mid 1\leq i < j \leq n, w_iw_j \in \mathcal{D}\}$. Note that $G(\mathcal{D}, w)$ can be seen as a graph with inherent coloring $\chi \colon V(G) \rightarrow \Sigma$.
It is unknown whether the lettericity of a given graph can be computed in polynomial time. The problem to determine the lettericity of a given graph is called the \emph{lettericity problem}. As a step towards answering the complexity of this problem, we investigate the following \emph{retrieval problems}: given a graph $G$ together with two of the three solution-objects (word $w$, decoder~$\mathcal{D}$, and coloring~$\chi$), the goal is to compute the third solution-object.
We show that word retrieval and decoder retrieval are solvable in polynomial time, while coloring retrieval is equivalent to the graph isomorphism problem. 
Beyond this, we introduce \emph{symmetric lettericity} which is a restricted version of lettericity where each decoder needs to be symmetrical ($ab\in \mathcal{D}$ if and only if~$ba\in \mathcal{D}$).
As we show, the symmetric lettericity of a graph always equals the neighborhood diversity of the graph, which in fact can be computed in linear time.

\end{abstract}

\section{Introduction}

For a given alphabet $\Sigma$, a word $w = w_1 \dots w_n \in \Sigma^*$ and a \emph{decoder}~$\Dmc \subseteq \Sigma^2$, we define the \emph{letter graph} $G(\Dmc, w)$ as the graph with the vertex set $\{1, \dots, n\}$ and the edge set~$\{ij \mid 1\leq i < j \leq n, w_iw_j \in \Dmc\}$. The \emph{lettericity} of a graph $G$ is defined as the smallest size of an alphabet $\Sigma$ such that there is a word $w\in \Sigma^*$ and a decoder $\Dmc \subseteq \Sigma^2$, such that $G$ is isomorphic to $G(\Dmc, w)$. Lettericity was introduced by Petkovšek in 2000~\cite{Pet2002}. For instance, for every $m\geq 2$, the star graph~$S_m$ with~$m$ leaves has lettericity $2$. This can be certified by the word $w := ab^m$ and the decoder $\Dmc := \{ab\}$. Note that, $S_m$ does not have lettericity $1$ because only complete and edge-less graphs have lettericity $1$. The \emph{lettericity problem} is the problem of determining the lettericity of a given graph $G$. 

In their original work, Petkovšek gave bounds or precise values for the lettericity of some graph classes such as threshold graphs, cycles and paths. For paths, the exact value was determined to be within one of two values depending on the length of the path. 
The problem of determining the lettericity of paths precisely remained open for approximately 20 years and was solved by Ferguson in 2020~\cite{Fer2020}.
 This goes to show that determining the exact lettericity even of relatively simple graph classes is far form trivial. 
 Petkovšek already proved in 2000 that for each constant $k$, the recognition problem for $k$-letter graphs, that is, graphs with lettericity at most~$k$, is solvable in polynomial time because the set of $k$-letter graphs is well-quasi-ordered and has a finite set of minimal forbidden subgraphs~\cite{Pet2002}. 
 In 2024, Alecu et al.~\cite{AleKLZ2024} improved upon this by showing that determining if a given graph has lettericity at most $k$ is FPT for the parameter $k$ using MSO-formulas. However the complexity of the lettericity problem remained open and they conjectured it to be NP-complete. 
For more literature on lettericity, see~\cite{AleFKLVZ2022,AleLoz2021,AtmCLZ2015,BerMah2016a,FerVat2022,ManVat2023}. Furthermore, Feng et al. introduced a generalization of lettericity~\cite{FenFMRS2025}, based on a generalized definition of decoders.

A graph parameter related to lettericity is \emph{thinness}, which was introduced by Mannino and Oriolo~\cite{ManOri2002}. They define the \emph{thinness} of a graph $G$ as the smallest number $k$ such that there exists an ordering of the vertices $v_1 < \dots < v_{|G(V)|}$ and a partition of $V(G)$ into $k$ classes $V_1, \dots, V_k$ such that for any triple $(r,s,t)$ with $r < s < t$, if $v_r$ and $v_s$ belong to the same class and $v_t v_r \in E(G)$, then $v_t v_s \in E(G)$.  
The \emph{thinness problem}, that is the problem of determining the thinness of a given graph, was proven to be \NP-complete by Shitov in 2025~\cite{Shi2025}. 
They used a result by Bonomo and de Estrada~\cite{BonEst2019}, who showed that the thinness problem with a given partition is \NP-complete. 
Shitov reduced the problem with a given partition to the general thinness problem, which implies that determining the thinness of a graph is NP-hard.
In 2026, Feng et al.~\cite{FenFFKS2026} observed that any $k$-letter graph is a $k$-thin graph. In their proof they used the inherent coloring $\chi\colon V(G) \to \Sigma$ as a partition. 
This motivates the idea that it might be possible to prove a complexity lower bound or upper bound for the lettericity problem with a similar approach, that is, studying the lettericity problem with a given part of the solution already fixed. 
In this work we call this kind of problem a \emph{retrieval problem}: given a graph $G$ together with two of the three \emph{solution objects}, word $w$, decoder~$\Dmc$, and coloring~$\chi$, the goal is to compute the third solution object. Indeed, we show that the retrieval of the coloring~$\chi$ is \GI-complete. Furthermore, we show that retrieving the word or the decoder can be done in polynomial time. 
Whether there exist reductions between the lettericity problem and these retrieval problems that would clarify the complexity of the lettericity problem remains open.
Finally, we also show that a version of the lettericity problem where the decoder has to be symmetric (that is, where $ab\in \Dmc$ if and only if~$ba\in \Dmc$) can be solved in polynomial time. 
Hence, if the lettericity problem would not be solvable in polynomial time, then a respective hardness proof would need to explicitly construct instances for which no decoder can be symmetrical.

Proofs of statements marked with~$(\star)$ are omitted due to space constraints.
\section{Preliminaries}

We denote by $\Nbb$ the set of positive integers (excluding 0) and define $[n] = \{1, \dots, n\}$ for every $n \in \Nbb$.
For a set~$X$, we denote by~$\binom{X}{2}$ the collection of all size-2 subsets of~$X$.

\newcommand{\emptyword}{\varepsilon}

\subsection{Words and Languages}
Let~$\Sigma$ be an alphabet, that is, a finite non-empty set and denote by~$\Sigma^*$ the set of all finite words over~$\Sigma$. A language is a subset $L \subseteq \Sigma^*$.
For a word $w \in \Sigma^*$, we denote by $|w|$ the length of~$w$ and by $w_i$ the $i$-th symbol in $w$. The word of length~$0$ is called \emph{empty word} and denoted by $\emptyword$.
We say a word $v \in \Sigma^*$ is a \emph{subsequence} of a word $w \in \Sigma^*$, denoted by $v \preccurlyeq w$, if $v$ can be obtained from~$w$ by removing symbols at arbitrary positions of~$w$. Formally, $v = v_1 \dots v_m$ is a subsequence of $w = w_1 \dots w_n$ if there is a function $f \colon [m] \rightarrow [n]$ such that $f(i) < f(j)$ if $i < j$ and $v_i = w_{f(i)}$ for all $i \leq m$. For a word $w = w_1 \dots w_n$, $y$ is a \emph{factor} of $w$ if $y = w_i \dots w_j$ for some $1 \leq i \leq j \leq n$. 
For a word~$w$ and a letter~$c$, we call a maximal non-empty factor of~$w$ that only contains the letter~$c$ a~\emph{$c$-run}.
For a set of letters~$\Sigma'\subseteq \Sigma$ and a word~$w\in \Sigma^*$, we denote by~$w[\Sigma']$ the unique maximal subsequence of~$w$ that only contains letters of~$\Sigma'$.
For sets of constant size, we may omit the set-brackets in this notation, that is, we may for example write~$w[a,b]$ instead of~$w[\{a,b\}]$.

\subsection{Graph Theory}

Throughout this paper mainly we consider undirected, simple, loopless graphs.
A graph~$G=(V,E)$ consists of its vertex set~$V(G):= V$ and edge set $E(G) := E$, where $E$ is a subset of all size-2 subsets of~$V$. 
For vertex sets~$A$ and~$B$, we denote by~$E(A,B)$ the set of all edges of~$E$ that have one endpoint in~$A$ and one endpoint in~$B$.
We define the (open) neighborhood of a vertex $v \in V$ by $N(v) = \{u \in V \mid \{u,v\} \in E\}$. Similarly, we define the closed neighborhood of $v$ by $N[v] = N(v) \cup \{v\}$. We call a pair of distinct vertices $u, v \in V$ \emph{(true) twins} if $N[u] = N[v]$ and \emph{false twins} if $N(u) = N(v)$.
Moreover, we call two vertices~$u$ and~$v$~\emph{generalized twins}, if~$N(u)\setminus \{v\} = N(v)\setminus \{u\}$.
Note that~$u$ and~$v$ are generalized twins if and only if they are true twins or false twins.
Given a subset $U \subseteq V$ of vertices, we denote by $G[U]$ the graph induced by $U$, that is, the graph with vertex set $U$ and edge set $\{\{u,v\} \in E \mid u, v \in U\}$. A \emph{graph isomorphism} between two graphs $G$ and $G'$ is a bijection $f\colon V(G) \to V(G')$ such that for each $u,v \in V(G)$ we have $\{u,v\} \in E(G)$ if and only if $f(u)f(v) \in E(G')$. An instance of the \textsc{Graph Isomorphism} problem consists of two graphs $G_1, G_2$, and the goal is to decide whether there exists a graph isomorphism between $G_1$ and $G_2$. Recall that \textsc{Graph Isomorphism} is $\GI$-complete by definition.

A \emph{permutation} is a sequence of the elements of a finite set $X$ that contains each element from $X$ exactly once. In this sense a permutation of~$X$ can be seen as a word over the alphabet $X$. 
A directed graph~$H=(V,A)$ consists of a set of vertices and a set of arcs~$A\subseteq V\times V$.
A \emph{topological ordering} of a directed graph $H=(V,A)$ is a permutation of $V$ satisfying~$(v,u)\notin A$ if~$u$ precedes~$v$. 
A~\emph{directed acyclic graph (DAG)} is a directed graph without cycles.
A directed graph is a DAG if and only if it has a topological ordering.

\subsection{Letter Graphs}

Let $\Sigma$ be an alphabet of size $k$. We call a language $\Dmc \subseteq \Sigma^2$ a \emph{decoder}. For a word $w = \Sigma^*$ and a decoder $\Dmc$, we define the \emph{$k$-letter graph} $G(\Dmc, w)$ to be the graph with vertex set $V(G(\Dmc, w)) = [|w|]$ and edge set $E(G(\Dmc, w)) = \{ij \mid 1 \leq i < j \leq |w|,  w_iw_j \in \Dmc\}$.
The \emph{lettericity} of a graph $G$ is the smallest value $k$ such that $G$ is isomorphic to a $k$-letter graph. Note that every letter graph inherently comes with a vertex coloring $\chi_w \colon V(G(\Dmc, w)) \rightarrow \Sigma$ defined as $\chi_w(i) = w_i$ for each $i \in [|w|]$. 
For each graph $G=(V,E)$ isomorphic to a letter graph $G(\Dmc, w)$, we say that the graph isomorphism  $f \colon V \to V(G(\Dmc, w))$ \emph{respects the coloring} $\chi\colon V \to \Sigma$ if $\chi(v) = w_{f(v)}$. Note that this is well-defined because $V(G(\Dmc, w)) = [|w|]$.

\section{Retrieval Problems}
In this section we study \emph{retrieval problems} which are of the following form. Given a graph $G=(V,E)$ together with two of the three objects, word $w \in \Sigma^*$, decoder $\Dmc \subseteq \Sigma^2$, and coloring $\chi\colon V \rightarrow \Sigma$, the goal is to compute the third object such that~$G$ is isomorphic to $G(\Dmc, w)$ respecting the coloring. We call an instance of any retrieval problem \emph{consistent} if such a third object exists.

\subsection{Word Retrieval}
\newcommand{\WR}{\textsc{Word Retrieval}\xspace}
The \WR problem is defined as follows. 
Given a graph $G=(V,E)$ with coloring $\chi \colon V \rightarrow \Sigma$ and a decoder $\Dmc \subseteq \Sigma^2$, compute a word $w$ such that $G$ is isomorphic to $G(\Dmc, w)$ respecting the coloring or report that no such word exists.

In the following, let~$I=(G=(V,E),\Sigma,\chi,\Dmc)$ be an instance of~\WR.
We define a directed auxiliary graph~$H=(V,A)$ as follows:
Initially, the arc set~$A$ is empty.
For each~$(u,v)\in V \times V$ with~$u\neq v$, we add the arc~$(u,v)$ to~$A$ if (i)~$\{u,v\}\in E$ and~$\chi(v)\chi(u)\notin \Dmc$ or (ii)~$\{u,v\}\notin E$ and~$\chi(v)\chi(u)\in \Dmc$. 
An example for this construction is given in~\Cref{fig WR}.

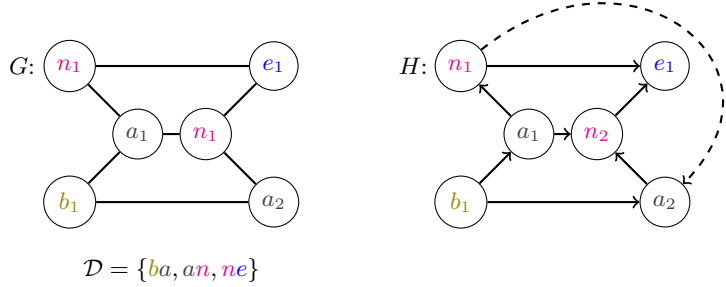
\begin{figure}[t!]
\centering
\begin{tikzpicture}[scale = .9]
	\newcommand{\colN}{magenta}
	\newcommand{\colE}{blue}
	\newcommand{\colA}{darkgray}
	\newcommand{\colB}{olive}
	\newcommand{\offset}{4.75}

		\node at (-0.7,2) {$G$:};
	
	\node[shape=circle,draw=black] (l_n_0) at (0,2) {\color{\colN}$n_1$};
	\node[shape=circle,draw=black] (l_e_0) at (3,2) {\color{\colE}$e_1$};
	\node[shape=circle,draw=black] (l_a_0) at (1,1) {\color{\colA}$a_1$};
	\node[shape=circle,draw=black] (l_n_1) at (2,1) {\color{\colN}$n_1$};
	\node[shape=circle,draw=black] (l_b_0) at (0,0) {\color{\colB}$b_1$};
	\node[shape=circle,draw=black] (l_a_1) at (3,0) {\color{\colA}$a_2$};
	
	\draw[thick] (l_n_0) -- (l_e_0);
	\draw[thick] (l_n_0) -- (l_a_0);
	\draw[thick] (l_e_0) -- (l_n_1);
	\draw[thick] (l_a_0) -- (l_n_1);
	\draw[thick] (l_a_0) -- (l_b_0);
	\draw[thick] (l_n_1) -- (l_a_1);
	\draw[thick] (l_b_0) -- (l_a_1);
	
	\node at (1.5,-1) {$\mathcal{D} = \{{\color{\colB}b}{\color{\colA}a},{\color{\colA}a}{\color{\colN}n},{\color{\colN}n}{\color{\colE}e}\}$};
	
	\begin{scope}[xshift = 1cm]
		\node at (\offset-0.7,2) {$H$:};
	
	\node[shape=circle,draw=black] (r_n_0) at (\offset+0,2) {\color{\colN}$n_1$};
	\node[shape=circle,draw=black] (r_e_0) at (\offset+3,2) {\color{\colE}$e_1$};
	\node[shape=circle,draw=black] (r_a_0) at (\offset+1,1) {\color{\colA}$a_1$};
	\node[shape=circle,draw=black] (r_n_1) at (\offset+2,1) {\color{\colN}$n_2$};
	\node[shape=circle,draw=black] (r_b_0) at (\offset+0,0) {\color{\colB}$b_1$};
	\node[shape=circle,draw=black] (r_a_1) at (\offset+3,0) {\color{\colA}$a_2$};

	\draw[thick,color=black,->] (r_n_0) -- (r_e_0);
	\draw[thick,color=black,->] (r_a_0) -- (r_n_0);
	\draw[thick,color=black,->] (r_a_0) -- (r_n_1);
	\draw[thick,color=black,->] (r_n_1) -- (r_e_0);
	\draw[thick,color=black,->] (r_b_0) -- (r_a_0);
	\draw[thick,color=black,->] (r_b_0) -- (r_a_1);
	\draw[thick,color=black,->] (r_a_1) -- (r_n_1);
	\draw[thick,color=black,dashed,->] (r_n_0) .. controls (\offset+2.5,4) and (\offset+5,2) .. (r_a_1);
	\end{scope}
\end{tikzpicture}
\caption{An example for the transformation of the \WR-instance to a directed graph. 
Here, a vertex~$c_i$ implies that this vertex receives color~$c$ under the labeling~$\chi$.
The solid arcs are added due to Condition (i) of the definition of the arc set and the single dashed arc is added due to Condition (ii).
The unique topological ordering~$b_1a_1n_1a_2n_2e_1$ of this example graph is the unique generalized solution for the~\WR-instance and implies that the word~$banane$ is a solution for the~\WR-instance.}
\label{fig WR}
\end{figure}

The arcs of $H$ describe the partial order that the nodes of $G$ must have in relation to each other in all words that are a solution to $I$.

We will show that we are dealing with a yes-instance of~\WR if and only if~$H$ is a DAG. To this end, note that we are dealing with a yes-instance of~\WR if and only if there is a \emph{generalized solution}, that is, a permutation~$\pi(1)\cdot \ldots \cdot \pi(n)$ of the vertices, such that~$\chi(\pi(i))\cdot \ldots \cdot \chi(\pi(n))$ is a solution for our~\WR-instance.

\begin{lemma}
Let~$\pi$ be a permutation of~$V$.
Then, $\pi$ is a topological ordering of~$H$ if and only if~$\pi$ is a generalized solution for~$I$. 
\end{lemma}
\begin{proof}
$(\Rightarrow)$
We show that for any~$1\leq i < j \leq |V|$, $\{u,v\}\in E$ if and only if $\chi(u)\chi(v)\in \Dmc$, where~$u:=\pi[i]$ and~$v:=\pi[j]$.

Assume~$\{u,v\}\in E$. 
If $\chi(u)\chi(v)\notin \Dmc$, then by definition of the arc set~$A$, $(v,u)\in A$.
As~$\pi$ is a topological ordering of~$H$, and~$u$ precedes~$v$ in~$\pi$, $(v,u)$ is not an arc of~$A$, which implies that~$\chi(u)\chi(v)\in \Dmc$. 

Assume~$\{u,v\}\notin E$. 
If $\chi(u)\chi(v)\in \Dmc$, then by definition of the arc set~$A$, $(v,u)\in A$.
As~$\pi$ is a topological ordering of~$H$, and~$u$ precedes~$v$ in~$\pi$, $(v,u)$ is not an arc of~$A$, which implies that~$\chi(u)\chi(v)\notin \Dmc$. 

Hence, $\{u,v\} \in E$ if and only if $\chi(u)\chi(v) \in \Dmc$ and $\pi$ is a generalized solution for~$I$.

$(\Leftarrow)$
We show that for any~$1\leq i < j \leq |V|$, $(v,u)$ is not an arc of~$A$, for~$u:=\pi[i]$ and~$v:=\pi[j]$.
 
Assume towards a contradiction that~$(v,u)\in A$.
Then, by definition of~$A$, (i)~$\{u,v\}\in E$ and~$\chi(u)\chi(v)\notin \Dmc$ or (ii)~$\{u,v\}\notin E$ and~$\chi(u)\chi(v)\in \Dmc$.

As~$\pi$ is a solution for~$I$, $\{u,v\}\in E$ if and only if~$\chi(u)\chi(v)\in \Dmc$.
This contradicts the assumption that~$(v,u)$ is an arc of~$A$.
Consequently, $\pi$ is a topological ordering of~$H$.
\lncsqed
\end{proof}

As~$H$ can be constructed in polynomial time and deciding whether a topological ordering of a directed graph exists can be done in linear time (and, if one exists, a topological ordering can also be found in that time), we obtain the following.

\begin{theorem}
\WR can be solved in polynomial time.
\end{theorem}

\subsection{Decoder Retrieval}

\newcommand{\DR}{\textsc{Decoder Retrieval}\xspace}

An instance~$I$ of \DR consists of a graph~$G= (V,E)$, an alphabet~$\Sigma$, a coloring~$\chi\colon V \to \Sigma$, and a word~$w\in \Sigma^{|V|}$.
For each letter~$a\in \Sigma$, we denote by~$V_a$ the vertices of color~$a$, that is, $V_a := \chi^{-1}(a)$.
We further assume that for each letter~$a\in \Sigma$, $V_a \neq \emptyset$ and that~$w$ contains the letter~$a$ exactly $|V_a|$~times.
The goal is to decide whether there is some~$\Dmc\subseteq \Sigma^2$, such that~$G$ can be obtained from~$\chi$ and $w$, by using~$\Dmc$ as the decoder.
We call each such~$\Dmc\subseteq \Sigma^2$ a~\emph{solution} (for~$I$).
If the vertex set of the graph is clear, we may simply define the edge set of the graph in an instance of~\DR.
That is, we may say that~$(E,\Sigma,\chi,w)$ is an instance of~\DR, where~$E$ is a set of edges.

We show that~\DR can be solved in polynomial time. To present the algorithm, we proceed in four steps. In the first step, we show that we can decide whether a given decoder~$\Dmc$ is in fact a solution for our given~\DR-instance.
Afterwards, we show that we can make the choice on whether~$ab$ or~$ba$ is in a solution greedily for many letter pairs~$\{a,b\}$.
The only cases where we cannot do this are pairs~$\{a,b\}$ where~$w[a,b]$ is a palindrome and edges between~$V_a$ and~$V_b$ are non-trivial, that is, some possible edges between these vertex sets exists and some do not exist.
For these remaining letter pairs~$\{a,b\}$ and~$\{b,c\}$ where the choice cannot be made trivially, a choice for the decoder word of~$\{ab,ba\}$ in a solution immediately forces a choice for the decoder word of~$\{cb,bc\}$ in the same solution.
Based on this, we finally reduce the decision on whether there is a solution for the~\DR instance to the satisfiability of a computable 2-SAT formula (after some sanity checks).

\begin{lemma}\label{lem verify decoder}
Let~$I=(G,\Sigma,\chi,w)$ be an instance of~\DR.
Moreover, let~$\Dmc\subseteq \Sigma^2$.
Then, in polynomial time, we can decide whether~$D$ is a solution for~$I$.
\end{lemma}

\begin{algorithm}[t!]
\DontPrintSemicolon
\KwIn{Graph $G=(V,E)$, an alphabet~$\Sigma$, a coloring $\chi \colon V \rightarrow \Sigma$, word $w = w_1 \dots w_{|V|} \in \Sigma^*$, and~$\Dmc \subseteq \Sigma^2$}
\KwOut{YES if and only if~$\Dmc$ is a solution for the instance~$(G,\Sigma, \chi, w)$ of \DR.}

\If{$|w| = 0$}{\Return YES}

{$a\cdot w' \gets w$}

\For{$b \in \Sigma$}{
	$N_b \gets \begin{cases}V_b& ab\in \Dmc\\\emptyset&ab\notin \Dmc \end{cases}$ 
}

\If{$\exists v\in V_a: N(v) = (\bigcup_{b\in \Sigma} N_b) \setminus \{v\}$\label{if condition}}
{\Return VerifyDecoder($G-v,\Sigma, \chi|_{V\setminus \{v\}},w',\Dmc$)}
	\Else{\Return NO}

\caption{VerifyDecoder($G,\Sigma, \chi,w,\Dmc$)}
\label{algo verify}
\end{algorithm}
\begin{proof}
We show that \Cref{algo verify} solves this task.
Clearly, the running time of the algorithm is polynomial.
Intuitively, the algorithm makes a choice for the vertex which will be isomorphic to the first vertex of the letter graph~$G(\Dmc,w)$.
This is then done recursively.
Since without loss of generality the first letter of~$w$ is~$a$, the first vertex~$v$ needs to be from~$V_a$ and its neighborhood~$N_b$ towards any vertex of~$V_b$ with~$b\in \Sigma$ has to be either~$V_b \setminus \{v\}$ or~$\emptyset$.
Here, we need to remove~$v$ from~$V_b$, as for~$b=a$, $v$ never has an edge towards itself, and for~$b\neq a$, $V_b \setminus \{v\} = V_b$.
If no such vertex~$v$ exists, the algorithm correctly outputs NO.
That is, if~$\Dmc$ is not a solution for the \DR-instance, the algorithm is correct.
It thus remains to show that the algorithm is also correct if~$\Dmc$ is a solution.
By the above, \Cref{if condition} has to hold for some vertex~$w\in V_a$.
Now assume towards a contradiction that the algorithm chose the wrong vertex~$v\in V_a$ at this line.
As both~$v$ and~$w$ fulfill the condition of \Cref{if condition}, we derive that~$v$ and~$w$ are generalized twins.
That is, $v$ and~$w$ are completely indistinguishable.
Hence, $v$ cannot be a wrong choice, if~$w$ would be a right choice.
Consequently, the algorithm correctly answers YES if~$\Dmc$ is a solution.
\lncsqed
\end{proof}

In the remainder of the section, we assume that we are given an instance~$I$ of~\DR and that in this instance, each color of~$\Sigma$ is assigned to at least one vertex.
We now continue with our sanity checks and greedy choice for some letter pairs.
The first observation considers the choice on~$\Dmc \cap \{aa\}$ for each letter~$a$.

\begin{observation}\label{color internal}
Let~$a\in \Sigma$.
If~$V_a$ is a clique in~$G$, then for every solution~$\Dmc$, $\Dmc \cup \{aa\}$ is a solution.
If~$V_a$ is an independent set in~$G$, then for every solution~$\Dmc$, $\Dmc \setminus \{aa\}$ is a solution.
If~$V_a$ is neither a clique nor an independent set, then we are dealing with a no-instance.  
\end{observation}

The next observation considers the choice on~$\Dmc \cap \{ab,ba\}$ for each letter pairs~$\{a,b\}$ where either all or no edges exists between~$V_a$ and~$V_b$.

\begin{observation}\label{identical between colors}
Let~$\{a,b\}\subseteq \Sigma$ with~$a\neq b$.
If~$|E(V_a,V_b)| = |V_a|\cdot |V_b|$, then for every solution~$\Dmc$, $\Dmc \cup \{ab,ba\}$ is a solution.
If~$|E(V_a,V_b)| = 0$, then for every solution~$\Dmc$, $\Dmc \setminus \{ab,ba\}$ is a solution.  
\end{observation}

We now give a name to letter pairs for which none of the above greedy choices applies.

\begin{definition}
We call an unordered letter pair~$\{a,b\}$ with~$a\neq b$~\emph{one-sided}, if~$0 < |E(V_a,V_b)| < |V_a|\cdot |V_b|$. 
\end{definition}
 
In the following, let~$\mpp$ denote the set of all one-sided letter pairs.

\begin{lemma}

\label{one sided}
Let~$\{a,b\}$ be a one-sided letter pair.
Then for each solution~$\Dmc$, $|\Dmc\cap \{ab,ba\}| = 1$.
\end{lemma}
\newcommand{\proofOneSided}{

\begin{proof}

Let~$\Dmc$ be a solution for~$I$.
Then, $|\Dmc \cap \{ab,ba\}| < 2$, as otherwise, each vertex of~$V_a$ has to be adjacent to each vertex of~$V_b$, which would contradict the fact that~$\{a,b\}$ is a one-sided letter pair.
Similarly, $|\Dmc \cap \{ab,ba\}| > 0$, as otherwise, each vertex of~$V_a$ has no neighbor in~$V_b$, which would contradict the fact that~$\{a,b\}$ is a one-sided letter pair.
Thus, $|\Dmc\cap \{ab,ba\}| = 1$.
\lncsqed
\end{proof}
}

\proofOneSided

Next, we show that one-sided letter pairs~$\{a,b\}$ for which~$w[a,b]$ contains only a single~$a$-run and a single~$b$-run guarantee that we are dealing with a no-instance.
Note that~$w[a,b]$ contains a single~$a$-run and a single~$b$-run if and only if~$w[a,b]\in a^+b^+\cup b^+a^+$.

\begin{lemma}
\label{sanity}
Let~$\{a,b\}$ be a one-sided letter pair, such that~$w[a,b]$ contains only one~$a$-run and only one~$b$-run, then~$I$ is a no-instance of~\DR.
\end{lemma}
\newcommand{\proofSanity}{

\begin{proof}

If~$w[a,b]$ contains only one~$a$-run and only one~$b$-run, then $w[a,b]$ is either~$a^xb^y$ or~$b^ya^x$ for values~$x,y\geq 1$.
Hence, for each decoder candidate~$\Dmc \subseteq \Sigma^2$, either all or none of the possible edges between vertices of~$V_a$ and~$V_b$ are implied.
As~$\{a,b\}$ is a one-sided letter pair, $0<|E(V_a,V_b)|<|V_a| \cdot |V_b|$, which implies that there there cannot be a solution for~$I$.
\lncsqed
\end{proof}
}

\proofSanity

In the following we thus gain insights for one-sided letter pairs~$\{a,b\}$ where~$w[a,b]$ contains at least two~$a$-runs or at least two~$b$-runs.
Additionally, we observe the following if~$w[a,b]$ is not a palindrome.

\newcommand{\figPalin}{
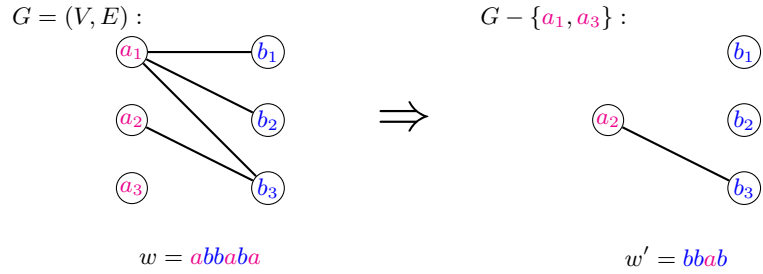
\begin{figure}[t!]
\newcommand{\colN}{magenta}
\newcommand{\colE}{blue}
\newcommand{\colA}{darkgray}
\newcommand{\colB}{olive}
\centering
\begin{tikzpicture}[scale=.9]

\begin{scope}[xshift=-3.5cm]
\node at (1,-2) {$w = {\color{magenta}a}{\color{blue}bb}{\color{magenta}a}{\color{blue}b}{\color{magenta}a}$};
\node at (-.8,1.5) {$G=(V,E):$};

\node[shape=circle,draw=black,inner sep = .5pt] (a1) at (0,1) {\color{magenta}$a_1$};
\node[shape=circle,draw=black,inner sep = .5pt] (a2) at (0,0) {\color{magenta}$a_2$};
\node[shape=circle,draw=black,inner sep = .5pt] (a3) at (0,-1) {\color{magenta}$a_3$};
\node[shape=circle,draw=black,inner sep = .5pt] (b1) at (2,1) {\color{blue}$b_1$};
\node[shape=circle,draw=black,inner sep = .5pt] (b2) at (2,0) {\color{blue}$b_2$};
\node[shape=circle,draw=black,inner sep = .5pt] (b3) at (2,-1) {\color{blue}$b_3$};

\draw[thick] (a1) to (b1);
\draw[thick] (a1) to (b2);
\draw[thick] (a1) to (b3);
\draw[thick] (a2) to (b3);

\end{scope}

\node at (.5,0) {\huge$\Rightarrow$};

\begin{scope}[xshift=3.5cm]
\node at (1,-2) {$w' = {\color{blue}bb}{\color{magenta}a}{\color{blue}b}$};
\node at (-.8,1.5) {$G - \{{\color{magenta}a_1},{\color{magenta}a_3}\}:$};

\node[shape=circle,draw=black,inner sep = .5pt] (a2) at (0,0) {\color{magenta}$a_2$};
\node[shape=circle,draw=black,inner sep = .5pt] (b1) at (2,1) {\color{blue}$b_1$};
\node[shape=circle,draw=black,inner sep = .5pt] (b2) at (2,0) {\color{blue}$b_2$};
\node[shape=circle,draw=black,inner sep = .5pt] (b3) at (2,-1) {\color{blue}$b_3$};

\draw[thick] (a2) to (b3);

\end{scope}

\end{tikzpicture}
\caption{An example for a word $w$ that is not a palindrome and the decoder can only contain one of the words~$ab$ or~$ba$.
The word~$w$ starts and ends with~$a$-runs of same length (namely, length 1), so there needs to be a vertex in~$V_a$ that is adjacent to all vertices of~$V_b$, and there needs to be a vertex in~$V_a$ that has no neighbor in~$V_b$.
These vertices are~$a_1$ and~$a_3$, respectively.
To derive the unique decoder word for~$\Dmc$, we can remove both vertices from the graph and the first and last~$a$-run in~$w$ (yielding the word~$w'=bbab$) and receive the instance depicted on the right.
Since~$w$ was not a palindrome, $w'$ is not a palindrome either.
Now, the first and last~$b$ run of~$w'$ have different lengths in the example: the first run has length~$2$ and the last one has length~$1$.
So there are either more vertices in~$V_b$ that are adjacent to all (remaining) vertices of~$V_a$ than vertices of~$V_b$ that have no neighbors to the (remaining) vertices of~$V_a$, or vide versa.
In this example, there are more vertices in~$V_b$ that have no neighbors, so the word~$ba$ cannot be in the decoder~$\Dmc$.
Hence, each decoder for~$w'$ (and thus for~$w$) needs to contain~$ab$.}
\label{fig palindrome}
\end{figure}
}

\figPalin

\begin{lemma}\label{compute non palindrome word}
Let~$\{a,b\}$ be a one-sided letter pair such that~$w[a,b]$ is not a palindrome.
We can compute in polynomial time a decoder word~$d_{ab}\in \{ab,ba\}$, such that for each solution~$\Dmc$, $d_{ab}\in \Dmc$.
\end{lemma}
\begin{proof}
We describe a constructive proof that shows that there is decoder word~$d_{ab}\in \{ab,ba\}$, such that for each solution~$\Dmc$, $d_{ab}\in \Dmc$.
Let~$\Dmc$ be a solution for~$I$.
Note that~$\Dmc' := \Dmc \cap \{ab,ba\}$ is a solution for~$I':=(E(V_a,V_b),\{a,b\},\chi|_{V_a\cup V_b}, w[a,b])$.
We thus analyze~$\Dmc'$ and show via induction over the length of~$w[a,b]$ that there is no solution for~$I'$ where~$\Dmc$ contains both, $ab$ and $ba$.
Without loss of generality we assume that $w[a,b]$ starts with an~$a$.

\textbf{Base case.} For a non-palindrome~$w[a,b]$ of length at most~$4$, the statement trivially holds. 

\textbf{Inductive step.} We consider two cases.

\textbf{Case 1:} $w[a,b]$ ends with letter~$b$\textbf{.}
If there is a vertex~$v_a\in V_a$ that is adjacent to all vertices of~$V_b$, then the word~$ab$ is in every solution for~$I'$.
If there is a vertex~$v_a\in V_a$ that is adjacent no vertex of~$V_b$, then the word~$ba$ is in every solution for~$I'$.
If neither of these is the case, there is no solution for~$I'$.

\textbf{Case 2:} $w[a,b]$ ends with letter~$a$\textbf{.}
Thus, $w[a,b]$ starts and ends with an~$a$.
Let~$R_f$ and~$R_\ell$ be the indices of the first and the last~$a$-run in~$w[a,b]$, respectively.
Moreover, let~$X_a$ denote the vertices of~$V_a$ that are adjacent to all vertices of~$V_b$.
Similarly, let~$Y_a$ denote the vertices of~$V_a$ that are adjacent to no vertex of~$V_b$. 
As~$I'$ is a yes-instance of~\DR and~$\{a,b\}$ is a one-sided letter pair, it holds that (1)~$|R_f| = |X_a|$ and~$|R_\ell| = |Y_a|$ or (2)~$|R_f| = |Y_a|$ and~$|R_\ell| = |X_a|$.

If~$|R_f| = |R_\ell|$, let~$V_a' := V_a \setminus (X_a\cup Y_a)$ and let~$w'$ be the word obtained from~$w[a,b]$ by removing the first and the last~$a$-run.
By the induction hypothesis, there is a word~$d_{ab}\in \{ab,ba\}$ such that~$d_{ab}$ is contained in every solution for the instance~$(E(V_a',V_b),\{a,b\},\chi|_{V_a'\cup V_b}, w')$ of~\DR.
This is due to the fact that~$w'$ is again not a palindrome.

Otherwise, that is, if~$|R_f| \neq |R_\ell|$, let without loss of generality~$|R_f| > |R_\ell|$.
Moreover, recall that this implies~$|X_a| \neq |Y_a|$, that is, there are more vertices of~$V_a$ that are universal for the vertices of~$V_b$ than vertices of~$V_a$ that have no neighbors in~$V_b$, or vice versa.
If~$|R_f| = |X_a|$, then each solution has to contain~$ab$, as the first~$a$-run has to contain all vertices of~$V_a$ that are universal for all vertices of~$V_b$.
If~$|R_f| = |Y_a|$, then each solution has to contain~$ba$, as the first~$a$-run has to contain all vertices of~$V_a$ that have no neighbors in~$V_b$.
\lncsqed
\end{proof}

In the next step, we show that the choice of the decoder word for~$\{a,b\}$ also determines the choice of the decoder word for~$\{b,c\}$ if both~$w[a,b]$ and~$w[b,c]$ contain at least two~$b$-runs each.

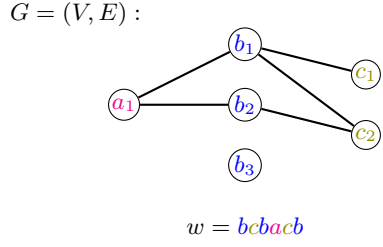
\begin{figure}[t!]
\centering
\begin{tikzpicture}[scale=.8]
\newcommand{\colN}{magenta}
\newcommand{\colE}{blue}
\newcommand{\colA}{darkgray}
\newcommand{\colB}{olive}

\begin{scope}[xshift=-3.5cm]
\node at (2,-2) {$w = {\color{blue}b}{\color{olive}c}{\color{blue}b}{\color{magenta}a}{\color{olive}c}{\color{blue}b}$};
\node at (-.8,1.5) {$G=(V,E):$};

\node[shape=circle,draw=black,inner sep = .5pt] (a1) at (0,0) {\color{magenta}$a_1$};
\node[shape=circle,draw=black,inner sep = .5pt] (b1) at (2,1) {\color{blue}$b_1$};
\node[shape=circle,draw=black,inner sep = .5pt] (b2) at (2,0) {\color{blue}$b_2$};
\node[shape=circle,draw=black,inner sep = .5pt] (b3) at (2,-1) {\color{blue}$b_3$};
\node[shape=circle,draw=black,inner sep = .5pt] (c1) at (4,.5) {\color{olive}$c_1$};
\node[shape=circle,draw=black,inner sep = .5pt] (c2) at (4,-.5) {\color{olive}$c_2$};

\draw[thick] (a1) to (b1);
\draw[thick] (a1) to (b2);
\draw[thick] (b1) to (c1);
\draw[thick] (b1) to (c2);
\draw[thick] (b2) to (c2);

\end{scope}

\end{tikzpicture}
\caption{An example for a word~$w$ where~$\{a,b\}$ and~$\{b,c\}$ are one-sided letter pairs, $w[a,b]$ and~$w[b,c]$ contain at least two $b$-runs each, and where~$w[b,c]$ is a palindrome.
The choice for the word of~$\{bc,cb\}$ for the decoder is determined by the choice for the word of~$\{ab,bc\}$ for the decoder.
Note that the latter is fixed, as~$w[a,b]$ is not a palindrome.
Namely, for each decoder~$\Dmc$ which is a solution, we need to have $\Dmc\cap \{ab,ba\} = \{ba\}$.
This implies that the first two~$b$s of~$w$ are associated with the vertices~$b_1$ and~$b_2$.
Consequently, solution~$\Dmc$ cannot fulfill~$\Dmc\cap \{bc,cb\} = \{cb\}$, as otherwise, $b_3$ (which by necessity would be associated with the last~$b$ in~$w$) would have a common edge with~$c_1$ in the respective letter graph~$G(\Dmc,w)$.} 
\label{fig cascade}
\end{figure}

\begin{lemma}\label{compute related word}
Let~$\{a,b\}$ and~$\{b,c\}$ be distinct one-sided letter pairs, such that~$w[a,b]$ and~$w[b,c]$ contain at least two $b$-runs each, and where~$w[b,c]$ is a palindrome.
Then for each~$d_{ab}\in \{ab,ba\}$, in polynomial time, we can compute~$d_{bc}\in \{bc,cb\}$, such that for each solution~$\Dmc$ with~$d_{ab}\in \Dmc$, $\Dmc$ also contains~$d_{bc}$. 
\end{lemma}
\begin{proof}
An example for this scenario is illustrated in~\Cref{fig cascade}.
Let~$d_{ab}\in \{ab,ba\}$.
First, we show that if there is a solution that contains~$d_{ab}$, then there is a unique~$d_{bc}\in \{bc,cb\}$, such that for each solution~$\Dmc$ with~$d_{ab}\in \Dmc$, $\Dmc$ also contains~$d_{bc}$.  
Assume towards a contradiction that there are two solutions~$\Dmc_1$ and~$\Dmc_2$ with~$\{d_{ab},bc\}\subseteq \Dmc_1$ and~$\{d_{ab},cb\}\subseteq \Dmc_2$.
For a letter pair~$p\in \{\{a,b\},\{b,c\}\}$, let~$R^p_f$ and~$R^p_\ell$ be the indices of the~$b$'s in~$w$ of the first and the last~$b$-run in~$w[p]$, respectively.
That is, if the length of the first~$b$ run in~$w[p]$ is~$\alpha$, then~$R^p_f$ contains the indices of the first~$\alpha$ many~$b$'s in~$w$.
Note that~$|R^{\{b,c\}}_f| = |R^{\{b,c\}}_\ell| < \frac{|V_b|}{2}$, as~$w[b,c]$ is a palindrome.
Assume without loss of generality that the length~$|R^{\{a,b\}}_f|$ of the first~$b$-run in~$w[a,b]$ is at most the length~$|R^{\{a,b\}}_\ell|$ of the last~$b$-run in~$w[a,b]$.
This implies that~$|R^{\{a,b\}}_f| < \frac{|V_b|}{2}$.
Consequently, $R^{\{a,b\}}_f \cap R^{\{b,c\}}_\ell = \emptyset$.

Under both~$\Dmc_1$ and~$\Dmc_2$, the vertices of~$V_b$ from~$G$ associated with the indices of~$R^{\{a,b\}}_f$ are the same, as both solutions contain only~$d_{ab}$ from~$\{ab,ba\}$.
Let~$X_b$ be this set of vertices of~$V_b$.
As~$\Dmc_1$ is a solution and~$bc\in \Dmc_1$, there is no vertex of~$X_b$ that has smallest degree towards the vertices of~$V_c$ among all vertices of~$V_b$, as~$R^{\{a,b\}}_f \cap R^{\{b,c\}}_\ell = \emptyset$.
However, since~$\Dmc_2$ is a solution and~$cb\in \Dmc_2$, the vertex associated to the first~$b$ (which is a vertex of~$X_b$) under~$\Dmc_2$ has the smallest degree towards the vertices of~$V_c$ among all vertices of~$V_b$.
This contradicts the assumption that both~$\Dmc_1$ and~$\Dmc_2$ are solutions.

It remains to show how to compute the respective word~$d_{bc}$ in polynomial time.
Consider the instance~$I':=(E(V_b,V_a\cup V_c), \{a,b,c\}, \chi|_{V_a\cup V_b\cup V_c}, w[\{a,b,c\}])$ of~\DR.
By the above arguments, at most one of~$\{d_{ab},bc\}$ and~$\{d_{ab},cb\}$ is a solution for~$I'$.
As we can construct~$I'$ in polynomial time, we can check in polynomial time (see~\Cref{lem verify decoder}) which of the two possible candidate solutions is in fact a solution for~$I'$.
If one of both is a solution, we output the respective word of~$\{bc,cb\}$.
Otherwise, that is if neither of both is a solution, then there is no solution that contains~$d_{ab}$, and we can correctly output any word from~$\{bc,cb\}$. 
\lncsqed
\end{proof}

For each letter~$a\in \Sigma$, let~$P_a$ denote the set of letters~$b$ of~$\Sigma\setminus \{a\}$ for which~$\{a,b\}$ is a one-sided letter pair and~$w[a,b]$ has at least two~$a$-runs.
Based on this definition, we characterize when a given decoder is a solution.

\begin{lemma}\label{equiv for valid decoder}
Let~$\Dmc\subseteq \Sigma^2$ and assume that for each~$\{a,b\}\in \mpp$, $w[a,b]$ contains at least two~$a$-runs or at least two~$b$-runs.
Then~$\Dmc$ is a solution for~$I$ if and only if 
\begin{itemize}
\item\label{solution cl vs is} for each letter~$a\in \Sigma$, $\Dmc \cap \{aa\}$ is a solution for the \DR-instance~$(E(V_a), \{a\}, \chi|_{V_a}, w[a])$, 
\item\label{solution twosided} for each~$\{a,b\}\in \binom{\Sigma}{2}\setminus \mpp$, $\Dmc \cap \{ab,ba\}$ is a solution for the \DR-instance~$(E(V_a,V_b), \{a,b\}, \chi|_{V_a\cup V_b}, w[a,b])$, and
\item\label{solution onesided} for each letter~$a\in \Sigma$, $\Dmc \cap \{ab,ba\mid b\in P_a\}$ is a solution for the~\DR-instance $(\bigcup_{b\in P_a} E(V_a,V_b), P_a \cup \{a\}, \chi|_{\cup_{b\in P_a\cup\{a\}}V_b}, w[a,P_a])$.
\end{itemize}
\end{lemma}
\begin{proof}
$(\Rightarrow)$
If~$\Dmc$ is a solution for~$I$, then all item need to hold, as they describe exactly the edge sets of subgraphs of~$G$.

$(\Leftarrow)$
Assume that all three items hold.
We show that~$\Dmc$ is a solution for~$I$.
Assume towards a contradiction that this is not the case.
Hence, \Cref{algo verify} outputs NO on input~$(G,\chi,w,\Dmc)$.
Let~$(G^*=(V^*,E^*),\Sigma,\chi^*,w^*,\Dmc)$ be the input of the recursive call at which the algorithm first outputs NO.
Moreover, let~$a$ be the first letter of~$w^*$ and let~$(v_a^1,\dots,v_a^p)$ be the sequence of vertices of~$V_a\setminus V^*$ in which they were chosen by previous calls of the algorithm.
Since the algorithm outputs NO, there is no vertex~$v_a^{p+1} \in V_a \cap V^*$, such that for each~$b\in \Sigma$, $N(v_a^{p+1}) \cap V_b \cap V^* = N_b \cap V^* \setminus \{v\}$.
To achieve a contradiction, we, however, show that such a vertex exists.
To this end, we first observe that for each~$b\in \Sigma \setminus P_a$, each vertex of~$V_a$ has the exactly same neighborhood towards the vertices of~$V_b$, that is, $N(v) \cap V_b \setminus \{w\} = N(w) \cap V_b \setminus \{v\}$ for each two vertices~$v,w\in V_a$.  
Here, recall that~$P_a$ contains all letters~$b$ for which $w[a,b]$ contains at least two~$a$-runs.

First, for~$b=a$, the statement holds, as~\Cref{solution cl vs is} implies that~$V_a = V_b$ is a clique or an independent set.
For~$b\in \Sigma$ with~$b\neq a$, such that~$\{a,b\}$ is not one-sided, the statement holds as \Cref{solution twosided} implies that~$G$ contains either all or no edge between~$V_a$ and~$V_b$.
Finally, consider a letter~$b\in \Sigma$, such that~$\{a,b\}$ is one-sided but~$b\notin P_a$.
The latter implies that~$w[a,b]$ contains only a single~$a$-run.
By the condition that there are no one-sided letter pairs that contain only a single run for each letter, this implies that~$w[a,b]$ contains at least/exactly two~$b$-runs.
Hence, $a\in P_b$.
Thus, \Cref{solution onesided} (for reversing the roles of~$a$ and~$b$) implies that there is a subset~$V_b'\subseteq V_b$, such that for each vertex~$v\in V_a$, $N(v)\cap V_b = V_b'$.
Consequently, for each letter~$b\in \Sigma \setminus P_a$, $N(v) \cap V_b \setminus \{w\} = N(w) \cap V_b \setminus \{v\}$ for each two vertices~$v,w\in V_a$.

Intuitively, this implies that the algorithm outputted NO, since there is no vertex~$v_a^{p+1}\in V_a \cap V^*$, such that for some~$b\in P_a$, $N(v_a^{p+1}) \cap V_b \cap V^* \neq N_b\cap V^*$.
However, since \Cref{solution onesided} holds, \Cref{algo verify} outputs YES on the input of \Cref{solution onesided}.
This implies that there is an ordering~$(w_a^1,\dots,w_a^{|V_a|})$ in which the algorithm selected the vertices of~$V_a$ in the recursive calls.
Note that~$w_a^i$ and~$v_a^i$ for~$1 \leq i \leq p$ are guaranteed generalized twins due to the observation that for all~$b\in \Sigma\setminus P_a$, all vertices of~$V_a$ have the same neighborhood towards~$V_b$.
Hence, we can further assume without loss of generality that~$w_a^i = v_a^i$ for each~$1 \leq i \leq p$, as generalized twins of the same color are indistinguishable.
This implies that $w_a^{p+1}$ is a vertex of~$V_a \setminus \{v_a^1,\dots, v_a^p\} = V_a \cap V^*$.
Hence, $N(w_a^{p+1}) \cap V_b\cap V^* = N_b \cap V^*$ for each~$b\in P_a$.
The contradicts the assumption that the algorithm outputted NO on input~$(G^*,\Sigma,\chi^*,w^*,\Dmc)$. 
Consequently, $\Dmc$ is a solution for~$I$.
\lncsqed
\end{proof}

We are now ready to describe our algorithm to solve~\DR.
To this end, we first make some sanity checks.
For each letter~$a\in \Sigma$, we check whether~$G[V_a]$ is a clique or an independent set.
If this is not the case, we deal with a no-instance and output this.
Similarly, we check for each one-sided letter pair~$\{a,b\}$, whether~$w[a,b]$ contains at least two~$a$-runs or at least two~$b$-runs.
If this is not the case, we also deal with a no-instance (see~\Cref{sanity}) and output this.
Note that all these checks can be done in polynomial time and that the latter sanity check implies that for each one-sided letter pair~$\{a,b\}$, $a\in P_b$ or~$b\in P_a$ holds.

If we have not outputted yet that we are dealing with a no-instance, we proceed as follows:
Recall that~$\mpp$ denotes the collection of all one-sided letter pairs.
Based on our previous observations, we describe a reduction of our problem to~\textsc{2-SAT}.
To this end, we describe a formula~$F$ over the variable set~$\{ab,ba\mid \{a,b\}\in \mpp\}$.
\begin{enumerate}
\item\label{clause one side} For each~$\{a,b\}\in \mpp$, we add the clauses~$ab \lor ba$ and $\neg ab \lor \neg ba$ that ensure that we chose exactly one of~$\{ab,ba\}$ for the decoder. 
\item\label{clause not palindrome} For each~$\{a,b\}\in \mpp$ where~$w[a,b]$ is not a palindrome, we add the clause~$d_{ab}$, where~$d_{ab}$ is the decoder word that is contained in every solution and that we can compute in polynomial time due to~\Cref{compute non palindrome word}.
\item\label{clause palindrome} For each two distinct $\{a,b\}\in \mpp$ and~$\{b,c\}\in \mpp$ where both~$w[a,b]$ and~$w[b,c]$ contain at least two~$b$-runs and where~$w[b,c]$ is a palindrome, do the following: 
For each~$d_{ab}\in \{ab,ba\}$, in polynomial time compute~$d_{bc}\in \{bc,cb\}$, such that for each solution~$\Dmc$ with~$d_{ab}\in \Dmc$, $\Dmc$ also contains~$d_{bc}$ (see~\Cref{compute related word}) and add the clause~$\neg d_{ab} \lor d_{bc}$.
\item\label{clauses blocks} For each~$a\in \Sigma$, let~$P_a^p$ be the subset of~$P_a$ of those letters~$b$ for which~$w[a,b]$ is a palindrome and let~$P_a^n := P_a\setminus P_a^p$.
We consider two cases in which we are interested in potential solutions for the subinstance $I_a:=(\bigcup_{b\in P_a} E(V_a,V_b), P_a \cup \{a\}, \chi|_{\cup_{c\in P_a\cup\{a\}}V_c}, w[P_a \cup\{a\}])$ of \DR.

\begin{enumerate}
\item\label{clause some non palindromes} If~$P_a^n \neq \emptyset$, let~$b\in P_a^n$.
For each~$c\in P_a^n$ let~$d_{ac}$ be the word of~$\{ac,ca\}$ outputted by the algorithm behind~\Cref{compute non palindrome word}.
Moreover, for each~$c\in P_a^p$, let~$d_{ac}$ be the word of~$\{ac,ca\}$ for the word~$d_{ab}$ outputted by the algorithm behind~\Cref{compute related word}.
If~$\{d_{ac}\mid c\in P_a\}$ is not a solution for~$I_a$, we add the clause~$\neg d_{ab}$ (which in fact leads to an unsatisfiable formula, as we also have the clause~$d_{ab}$ due to~\Cref{clause not palindrome}).

\item\label{clause all palindromes} If~$P_a^n = \emptyset$, let~$b\in P_a^p$.
For each~$d_{ab}\in \{ab,ba\}$, do the following:
For each~$c\in P_a^p\setminus \{b\}$, let~$d_{ac}$ be the word of~$\{ac,ca\}$ for the word~$d_{ab}$ outputted by the algorithm behind~\Cref{compute related word}.
If~$\{d_{ac}\mid c\in P_a\}$ is not a solution for~$I_a$, we add the clause~$\neg d_{ab}$.
\end{enumerate}
\end{enumerate}

To prove our polynomial time algorithm, we now show the following.

\begin{lemma}
It holds that~$I$ is a yes-instance of~\DR if and only if the~$2$-SAT formula~$F$ is satisfiable.
\end{lemma}
\begin{proof}
$(\Rightarrow)$
Let~$\Dmc$ be a solution for~$I$.
We show that setting exactly the variables of~$A:=\Dmc\cap \{ab,ba\mid \{a,b\}\in \mpp\}$ to True satisfies the formula~$F$.
As~$\Dmc$ is a solution and for each~$\{a,b\}\in \mpp$ is a one-sided letter pair, \Cref{one sided} implies that~$|\Dmc\cap \{ab,ba\}| = 1$.
Hence, the clauses of~\Cref{clause one side} are satisfied.
Further, \Cref{compute non palindrome word} implies that for each~$\{a,b\}\in\mpp$ where~$w[a,b]$ is not a palindrome, the word~$d_{ab}$ is part of~$\Dmc$.
Hence, the clauses of~\Cref{clause not palindrome} are satisfied.
Now, let $\{a,b\}\in \mpp$ and~$\{b,c\}\in \mpp$ be distinct letter pairs where both~$w[a,b]$ and~$w[b,c]$ contain at least two~$b$-runs and where~$w[b,c]$ is a palindrome.
Let~$d_{ab}$ be the unique word of~$\{ab,ba\}$ that is contained in~$\Dmc$.
By \Cref{compute related word}, there is a unique~$d_{bc}\in \{bc,cb\}$, such that~$d_{bc}\in \Dmc$.
Hence, the clauses of~\Cref{clause palindrome} are satisfied.
Finally, we consider the clauses of~\Cref{clauses blocks}.
To this end, let~$a\in \Sigma$ and we distinguish whether~$P_a^n = \emptyset$.

If~$P_a^n \neq \emptyset$, let~$b\in P_a^n$ be as chosen in~\Cref{clause some non palindromes}. 
As~$\Dmc$ is a solution, \Cref{compute non palindrome word} implies that for each~$c\in P_a^n$, $d_{ac}\in \Dmc$.
In particular, $d_{ab}\in \Dmc$.
Thus, \Cref{compute related word} implies that~$d_{ac}\in \Dmc$ for each~$c\in P_a^p$.
This in particular implies that~$\Dmc\cap \{ac,ca\mid c\in P_a\} = \{d_{ac}\mid c\in P_a\}$ is a solution for~$I'$.
Hence, the clause~$\neg d_{ab}$ was not added in~\Cref{clause some non palindromes}, which implies that the clauses of~\Cref{clause some non palindromes} for letter~$a$ are satisfied. 
 
If~$P_a^n = \emptyset$, let~$b\in P_a^p$ be as chosen in~\Cref{clause all palindromes}. 
Moreover, let~$d_{ab}$ be the unique word of~$\{ab,ba\}\cap \Dmc$.
As~$\Dmc$ is a solution, \Cref{compute related word} implies that~$d_{ac}\in \Dmc$ for each~$c\in P_a^p$.
This in particular implies that~$\Dmc\cap \{ac,ca\mid c\in P_a\} = \{d_{ac}\mid c\in P_a\}$ is a solution for~$I'$.
Hence, the clause~$\neg d_{ab}$ was not added in~\Cref{clause all palindromes}, which implies that the clauses of~\Cref{clause all palindromes} for letter~$a$ are satisfied. 

Thus, $F$ is satisfied by assigning exactly the variables of~$A$ to True.

$(\Leftarrow)$
Let~$A$ be exactly the variables assigned to True by an arbitrary but fixed satisfying assignment of~$F$.
We show that there is a solution~$\Dmc$ for~$I$ with~$\Dmc \cap \{ab,ba\mid \{a,b\}\in \mpp\} = A$.
In particular, we show that~$\Dmc:= A \cup B$ is a solution for~$I$, where~$B:= \{aa\mid a\in \Sigma, G[V_a]~\text{is a clique}\} \cup \{ab,ba\mid \{a,b\}\in \binom{\Sigma}{2} \setminus \mpp, |E(V_a,V_b)| = |V_a| \cdot |V_b|\}$.
As we only constructed the formula~$F$ after our sanity check in which we checked whether for each~$a\in \Sigma$, $G[V_a]$ is a clique or an independent set, \Cref{color internal} and~\Cref{identical between colors} imply that the choice of~$B$ satisfies the first two items of~\Cref{equiv for valid decoder}.
To show that~$\Dmc$ is a solution, it thus remains to show that also the last point of~\Cref{equiv for valid decoder} holds.
This follows immediately by the fact that the clauses of~\Cref{clause some non palindromes} and~\Cref{clause all palindromes} are satisfied by setting exactly the variables of~$A$ to True.
Hence, \Cref{equiv for valid decoder} implies that~$\Dmc$ is a solution for~$I$.
\lncsqed
\end{proof}

This thus implies the following, as the formula~$F$ can be computed in polynomial time and the previous sanity check also runs in polynomial time.

\begin{theorem}
\DR can be solved in polynomial time.
\end{theorem}

\subsection{Coloring Retrieval}
\newcommand{\CR}{\textsc{Coloring Retrieval}\xspace}
\newcommand{\GIP}{\textsc{Graph Isomorphism}\xspace}
The \CR problem is defined as follows. Given a graph $G=(V,E)$, a word $w \in \Sigma^*$ and a decoder $\Dmc \subseteq \Sigma^2$, decide whether there is a coloring $\chi \colon V \rightarrow \Sigma$ that is consistent with the inherent coloring of $G(\Dmc, w)$.
\begin{theorem}\label{gi hardness}
\CR is \GI-complete.
\end{theorem}
\begin{proof}
The containment in \GI{} holds immediately: An instance~$(G,\Sigma,\Dmc,w)$ of~\CR is a yes-instance if and only if $G$ and $G(\Dmc, w)$ are isomorphic.
As~$G(\Dmc, w)$ can be computed in polynomial time, there is a trivial reduction from~\CR to \GIP.

The~\GI-hardness can be seen as follows.
Let~$I:=(G=(V,E),H=(U,F))$ be an instance of~\GIP.
We obtain an instance~$I':=(G',\Sigma,\Dmc,w)$ of~\CR as follows:
We set~$G':= G$, $\Sigma := U$, $\Dmc := \{uv,vu\mid \{u,v\}\in F\}$, and define~$w$ to be an arbitrary permutation of~$U$. 
Note that~$H = G(\Dmc,w)$.
This immediately implies that~$I$ is a yes-instance if and only if~$G$ is isomorphic to~$G(\Dmc,w)$.
Hence, \CR is also \GI-hard.\lncsqed
\end{proof}

\section{Symmetric Lettericity}
\newcommand{\nd}{\mathrm{nd}}
\newcommand{\sll}{\mathrm{sl}}
In this section, we consider the lettericity problem, when each allowed decoder needs to be~\emph{symmetric}, that is, $ab\in \Dmc$ if and only if~$ba\in \Dmc$.

Let~$G=(V,E)$ be a graph.
We let~$\sll(G)$ denote the smallest size of any alphabet~$\Sigma$, such that there is a symmetric decoder~$\Dmc\subseteq \Sigma^2$ and a word~$w\in \Sigma^{|V|}$ for which~$G$ is isomorphic to~$G(\Dmc,w)$.
Note that~$G(\Dmc,w)$ is isomorphic to~$G(\Dmc,w')$ for every permutation~$w'$ of~$w$.

\begin{observation} \label{rem:lettericityAndSymmetricLettericity}
For each graph $G$ we have $l(G) \leq sl(G)$.\end{observation}

We show that~$\sll(G)$ equals the neighborhood diversity of~$G$, a parameter that can be computed in linear time.
Here, the neighborhood diversity~$\nd(G)$ of~$G$ denotes the smallest partition~$(V_1,\dots,V_p)$ of~$V$, such that for each~$i\in [1,p]$ and for any two vertices~$u,v\in V_i$, $u$ and~$v$ are generalized twins, that is $N(u)\setminus \{v\} = N(v)\setminus \{u\}$.
Note that this implies that for each~$i\in [1,p]$, $V_i$ is a clique or an independent set.  

We start by showing that two vertices that are not generalized twins cannot receive the same color under.

\begin{lemma}
\label{sl geq nd}
Let~$G=(V,E)$ be a graph and let~$u,v\in V$ with~$N(u)\setminus \{v\} \neq N(v)\setminus \{u\}$.
Then for every symmetric decoder~$\Dmc$ and each word~$w$ for which there is a graph isomorphism~$f\colon V \to V(G(\Dmc,w))$, $w_{f(u)} \neq w_{f(v)}$.
\end{lemma}
\newcommand{\proofSlGeqNd}{

\begin{proof}

Assume towards a contradiction that~$w_{f(u)} = w_{f(v)}$.
As~$N(u)\setminus \{v\} \neq N(v)\setminus \{u\}$, there is a vertex~$x\in V \setminus \{u,v\}$ that has exactly one neighbor in~$\{u,v\}$.
Assume without loss of generality that~$x$ is adjacent to~$u$ and non-adjacent to~$v$.
As~$f$ is a graph isomorphism between~$G$ and~$G(\Dmc,w)$, $\Dmc$ contains at least one of~$\{w_{f(u)}w_{f(x)}, w_{f(x)}w_{f(u)}\}$.
In fact, since~$\Dmc$ is a symmetric decoder, $\Dmc$ contains both these word.
Hence, by~$w_{f(u)} = w_{f(v)}$, we get that~$f(v)$ is adjacent to~$f(x)$ in~$G(\Dmc,w)$.
This contradicts the fact that~$f$ is a graph isomorphism between~$G$ and~$G(\Dmc,w)$, as~$v$ and~$x$ are non-adjacent in~$G$.
Thus, $w_{f(u)} \neq w_{f(v)}$.
\lncsqed
\end{proof}
}

\proofSlGeqNd

Note that this immediately implies that for each graph~$G$, $\sll(G) \geq \nd(G)$.
We now show that the inequality also holds in the other direction.

\begin{theorem}
\label{sl eq nd}
For every graph~$G$, $\sll(G) = \nd(G)$.
\end{theorem}

\newcommand{\proofSlEqNd}{

\begin{proof}

Due to~\Cref{sl geq nd}, it remains to show that~$\sll(G) \leq \nd(G)$.
Let~$p:= \nd(G)$ and let~$(V_1,\dots,V_p)$ be a partition of~$V$, such that for each~$i\in [1,p]$ and for any two vertices~$u,v\in V_i$, $N(u)\setminus \{v\} = N(v)\setminus \{u\}$.
Recall that this implies that for each~$i\in [1,p]$, $G(V_i)$ is a clique or an independent set.
Moreover, it implies that for all~$1\leq i < j \leq p$, $|E(V_i,V_j)| \in \{0, |V_i|\cdot|V_j|\}$.
That is, either each vertex of~$V_i$ is adjacent to each vertex of~$V_j$, or there are no edges between vertices of~$V_i$ and vertices of~$V_j$.

To show that~$\sll(G) \leq \nd(G)=p$, we define an alphabet~$\Sigma$ of size~$p$, a symmetric decoder~$\Dmc\subseteq \Sigma^2$, and a word~$w\in \Sigma^{|V|}$ such that~$G(\Dmc,w)$ is isomorphic to~$G$.
We set~$\Sigma := [1,p]$ and~$w:= 1^{|V_1|}\cdot \ldots \cdot p^{|V_p|}$.
Next, we define the symmetric decoder~$\Dmc$.
We initialize~$\Dmc$ as~$\{ii\mid i\in [1,p], G[V_i]~\text{is a clique}\}$.
For each~$1\leq i< j \leq p$, we add the words~$ij$ and~$ji$ to the decoder if and only if~$|E(V_i,V_j)| = |V_i|\cdot|V_j|$.
This completes the definition of~$\Dmc$.
We show that~$G$ is isomorphic to~$G(\Dmc,w)$.
To this end, let~$f\colon V \to V(G(\Dmc,w))$ be an arbitrary bijection fulfilling for each~$v\in V$ and~$i\in [1,p]$, $w_{f(v)} = i$ if and only if~$v\in V_i$.
Since~$w$ contains exactly $|V_i|$~many times the letter~$i$, such a function~$f$ exists.
We show that~$f$ is a graph isomorphism between~$G$ and~$G(\Dmc,w)$.

Let~$u$ and~$v$ be distinct vertices of~$V$. 
We show that~$\{u,v\}\in E$ if and only if~$\{f(u),f(v)\}$ is an edge of~$G(\Dmc,w)$.
Let~$i,j\in [1,p]$, such that~$u\in V_i$ and~$v\in V_j$.

\textbf{Case: $i=j$.} 
Hence, $\{u,v\}\in E$ if and only if~$V_i=V_j$ is a clique in~$G$.
This is the case if and only if~$ii\in \Dmc$, which is the case if and only if~$\{f(u),f(v)\}$ is an edge of~$G(\Dmc,w)$, as~$w_{f(u)} = w_{f(v)} = i$.

\textbf{Case: $i\neq j$.} 
Hence, $\{u,v\}\in E$ if and only if~$|E(V_i,V_j)| = |V_i|\cdot|V_j|$.
This is the case if and only if~$\{ij,ji\}\subseteq \Dmc$, which is the case if and only if~$\{f(u),f(v)\}$ is an edge of~$G(\Dmc,w)$, as~$w_{f(u)} = i$ and~$w_{f(v)} = j$.

Thus, $G$ is isomorphic to~$G(\Dmc,w)$, which implies that~$\sll(G) \leq p$ by the fact that~$\Dmc$ is a symmetric decoder and~$\Sigma$ has size~$p$. 
\lncsqed
\end{proof}
}

\proofSlEqNd

\section{Conclusion}
We now briefly discuss another direction to approach the complexity of the lettericity problem.
Namely, in contrast to asking whether we can derive/retrieve one of the three objects for a solution of the lettericity problem, we can ask for a single given such object, whether we can derive the other two objects.
We call the respective problems~\emph{extension problems}.
For example, \DE would be the problem, where we are given a graph~$G$, an alphabet~$\Sigma$, and a decoder~$\Dmc\subseteq \Sigma^2$, and we ask whether there is a word~$w$ and a coloring~$\chi$, such that~$G(\Dmc,w) = G$, while respecting the coloring~$\chi$.

For example, \DE can easily be observed to be \GI-hard.
This hardness follows from the reduction behind~\Cref{gi hardness}, as the chosen word~$w$ for the \CR-instance was chosen arbitrarily among all possible permutations of the alphabet.
Moreover, we can ensure that no two vertices in the \GIP-instance are generalized twins, since subdividing each edge in both~$G$ and~$H$ ensures this property (assuming there are no connected components of size at most~$2$) and thus yields an equivalent instance.

For future work, it would be interesting to analyze the other two extension problems, namely~\textsc{Word Extension}, where we try to find a decoder and a coloring for a given word, and~\textsc{Coloring Extension}, where we try to find a decoder and a word for a given coloring.
Our presented algorithms for the retrieval problems do not seem to be adaptable to work for these extension problems, as the algorithms highly relied on having two of the three objects being known.

Other interesting directions for future work could be to look into generalizations of the lettericity problem. For example, what if we allow each vertex to receive an arbitrary set of color instead of just a single one. That is, $w$ would not be a word of~$\Sigma^n$ anymore but rather a sequence of~$n$ subsets of~$\Sigma$.
For such a generalization, multiple possible rules for when an edge is added to the letter graph can be considered. 
For example, let~$i$ and~$j$ be positions in~$w$ with~$i<j$ and let~$C_i$ and~$C_j$ be the set of colors of these two positions respectively.
One way to define the edges of the letter graph could be to require that for each~$\alpha\in C_i$ and each~$\beta\in C_j$, $\alpha\beta$ is a word of the decoder~$\Dmc$.
Other edge definition rules could require that (i)~there is at least one~$\alpha\in C_i$ and at least one~$\beta\in C_j$, such that~$\alpha\beta\in \Dmc$ or (ii)~the number of pairs~$(\alpha,\beta)\in C_i\times C_j$ with~$\alpha\beta\in \Dmc$ is at least a given threshold~$t$ (for example~$t=\frac{|C_i\times C_j|}{2}$).
Such a generalization of letter graphs could allow for an even more compact representation of graphs.
Note that for the normal definition of~$w$, that is, where each vertex as only a single color, all three example definitions coincide.

\bibliographystyle{splncs04}
\bibliography{literature}

\end{document}